\documentclass[letterpaper, 10 pt, conference]{ieeeconf}  

\IEEEoverridecommandlockouts                              
\usepackage{graphicx}
\usepackage{amsmath}
\usepackage{amsfonts}
\usepackage{amssymb}
\usepackage{subcaption}
\usepackage{float}
\usepackage{algpseudocode}

\usepackage{multirow}
\usepackage{enumerate}
\usepackage{tikz}

\usepackage{pgfplots}   
\usepackage{booktabs}  
\usepackage{multirow}   
\usepackage{float}     

\usepackage{booktabs}
\usepackage{algorithm}

\usepackage{xcolor}
\usepackage{mathtools}
\usepackage{acronym}
\usepackage{bm}         
\usepackage{hyperref}

\usepackage{amsthm}
\theoremstyle{plain}
\newtheorem{theorem}{Theorem}[section]
\newtheorem{proposition}[theorem]{Proposition}

\theoremstyle{definition}
\newtheorem{definition}[theorem]{Definition}
\theoremstyle{remark}
\newtheorem{remark}[theorem]{Remark}
\theoremstyle{definition}
\newtheorem{example}{Example}
\theoremstyle{lemma}

\begin{document}

\title{Hybrid Polynomial Zonotopes: A Set Representation for Reachability Analysis in Hybrid Nonaffine Systems}

\author{Peng Xie$^{*}$, Zhen Zhang$^{*}$, Amr Alanwar
\thanks{$^{*}$Peng Xie and Zhen Zhang contributed equally to this work. 
Peng Xie, Zhen Zhang, and Amr Alanwar are with the TUM School of Computation, Information and Technology, Department of Computer Engineering, Technical University of Munich, Heilbronn, Germany. 
Email: \{p.xie, zhenzhang.zhang, alanwar\}@tum.de}
}

\maketitle

\begin{abstract}
Reachability analysis for hybrid nonaffine systems remains computationally challenging, as existing set representations—including constrained, polynomial, and hybrid zonotopes—either lose tightness under high-order nonaffine maps or suffer exponential blow-up after discrete jumps. This paper introduces \textbf{Hybrid Polynomial Zonotope (HPZ)}, a novel set representation that combines the mode-dependent generator structure of hybrid zonotopes with the algebraic expressiveness of polynomial zonotopes. HPZs compactly encode non-convex reachable states across modes by attaching polynomial exponents to each hybrid generator, enabling precise capture of high-order state-input couplings without vertex enumeration. We develop a comprehensive library of HPZ operations including Minkowski sum, linear transformation, intersection, and union. Theoretical analysis and computational experiments demonstrate that HPZs achieve superior tightness preservation and computational efficiency for hybrid system reachability analysis.
\end{abstract}

\section{Introduction}
\label{sec:introduction}

Hybrid nonlinear systems refer to hybrid systems where each continuous submodel exhibits nonaffine nonlinear dynamic relationships, meaning the relationships between states and control inputs are nonlinear and do not take affine forms~\cite{mcclamroch2002performance}. Such systems are prevalent in practical engineering applications, including robotic switching control~\cite{borquez2023hamilton,prieur2007hybrid,sun2011stability} and so on. Due to the difficulty in linearizing or simplifying nonaffine subsystems, traditional analysis and design methods developed for linear or affine systems face significant challenges when applied to these systems.

In hybrid system modeling approaches, switching system models consist of a finite number of continuous subsystems and switching rules, where the switching rules determine when the system transitions between different subsystems (modes). There are lots of related works: ~\cite{hespanha2004stability, liberzon2003switching,schwartz2005minimum,long2022robust, long2021dynamic}. Each subsystem is typically described by a set of differential equations. However, the modeling and analysis of hybrid nonlinear systems face numerous challenges. Traditional reachability analysis methods for hybrid systems primarily focus on cases where the continuous dynamics are linear for each discrete mode~\cite{bird2023hybrid}. This assumption severely limits applicability to real-world systems, as nonaffine dynamics are prevalent in practical engineering applications.


\begin{figure}[htbp]
    \centering
\includegraphics[width=0.5\textwidth]{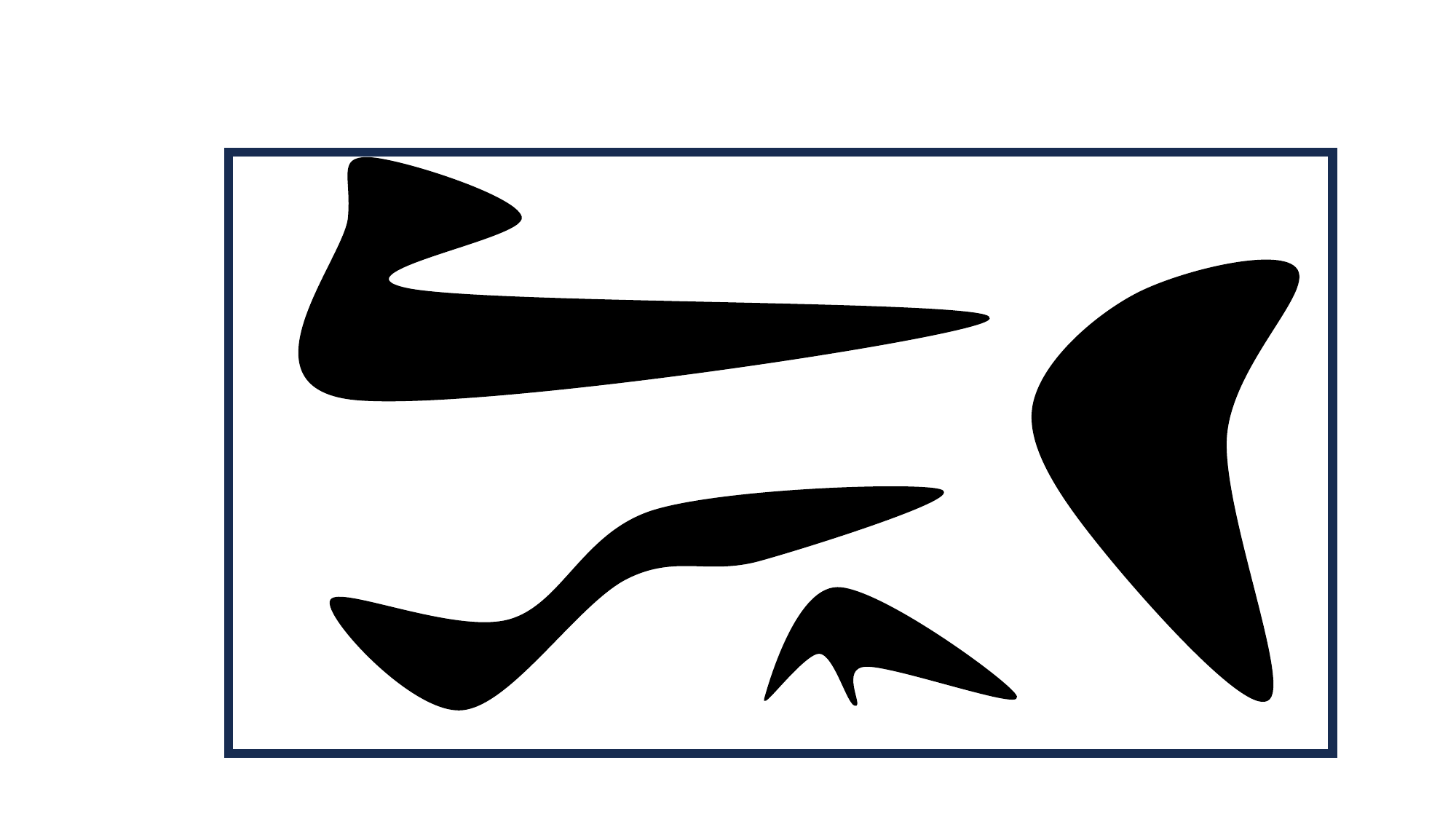}
    \caption{Four non-connected, non-convex obstacles (black regions) that require hybrid polynomial zonotope representation. Using 2-units binary variables, HPZ can efficiently encode these complex geometries, while individual constrained zonotopes and polynomial zonotopes cannot adequately represent such disconnected non-convex shapes.}
    \label{fig-Example0}
\end{figure}
The trade-off between modeling conservatism and accuracy represents a critical challenge. To analyze nonaffine systems using existing tools, models often require approximation and simplification, such as linearizing nonlinear systems~\cite{althoff2010reachability}, neglecting higher-order terms, or employing piecewise linear approximations~\cite{nikolaou2003linear, asarin2000reachability}. However, such approximation and simplification can introduce severe conservatism and fail to accurately reflect the original system behavior. Particularly in hybrid nonlinear systems, higher-order terms often contain crucial dynamic characteristic information~\cite{emara1996nonlinear, nonlinear_approximation}. Neglecting these higher-order terms may lead to omission of critical system behaviors or even introduce non-existent equilibrium points or spurious invariant sets. While such simplification may have limited impact in single-mode systems, in multi-mode switching hybrid systems, errors accumulate and amplify during mode transitions, ultimately causing severe deviations in reachability analysis results and compromising the reliability of safety verification.

Existing set representation methods—including constrained zonotopes~\cite{scott2016constrained}, polynomial zonotopes~\cite{kochdumper2020sparse}, and hybrid zonotopes~\cite{bird2023hybrid}—either suffer from failing to effectively handle the complexity of hybrid dynamics or lose tightness under high-order nonlinear mappings. The fundamental limitation lies in the inability to accurately represent arbitrary shapes, particularly non-convex, non-connected geometries arising in nonlinear systems. Polynomial zonotopes~\cite{kochdumper2020sparse} and constrained polynomial zonotopes~\cite{kochdumper2023} can represent connected non-convex shapes through polynomial generators, which enable exact forward propagation of nonconvex initial states in reachability analysis~\cite{zhang2025data}, but are restricted to single-mode continuous systems and cannot handle disconnected regions. Hybrid zonotopes~\cite{bird2023hybrid} can express non-connected shapes across multiple modes, yet each individual component remains essentially a constrained zonotope~\cite{scott2016constrained}, thus preserving convexity within each mode. As illustrated in Fig.~\ref{fig-Example0}, consider a workspace containing four distinct obstacles with complex geometries where each obstacle exhibits non-convex shapes and the obstacles are spatially disconnected from one another. Neither constrained polynomial zonotopes nor hybrid zonotopes can adequately capture such complex geometries—constrained polynomial zonotopes fail to represent the disconnected obstacles, while hybrid zonotopes cannot efficiently encode the polynomial-based curved boundaries that define each obstacle's non-convex geometry~\cite{hadjiloizou2024formal}.

To address these limitations, we propose a novel representation structure called \emph{Hybrid Polynomial Zonotopes (HPZ)}. Under the assumption that discrete logic is known, inspired by hybrid zonotope and its related applications~\cite{bird2023hybrid,koeln2024zonolab, siefert2025reachability}, we propose the HPZ method to precisely describe system states for piecewise nonaffine systems. HPZ combines the mode-dependent generator structure of hybrid zonotopes with the algebraic expressiveness of polynomial zonotopes. By attaching polynomial exponents to each hybrid generator, HPZ compactly encodes non-convex reachable states across modes, thereby capturing high-order state-input couplings without requiring vertex enumeration.

Our approach provides a unified modeling framework and analysis tool capable of precisely describing system behavior and addressing reachability analysis challenges for hybrid nonlinear systems. By employing HPZ, we enable more accurate reachability analysis for a broader class of hybrid systems where the dynamics within each discrete state can be nonaffine nonlinear—representing a significant advancement over traditional methods.

The main contributions of this paper are:
\begin{enumerate}
    \item The introduction of Hybrid Polynomial Zonotopes that extends polynomial zonotopes to hybrid systems with nonaffine submodels;
    \item A comprehensive library of HPZ operations including Minkowski sum, linear transformation, intersection operations, and union;
    \item Algorithms for performing reachability analysis using Hybrid Polynomial Zonotopes, including operations for continuous evolution and discrete transitions;
    \item The implementation of hybrid polynomial zonotopes and the associated algorithms publicly available as an open-source extension at \url{https://github.com/TUM-CPS-HN/HPZ}.
\end{enumerate}

The remainder of this paper is organized as follows: Section~\ref{sec-prelim} reviews related work on reachability analysis for hybrid systems and provides the necessary background on polynomial zonotopes and hybrid zonotopes. Section~\ref{sec-HPZ} introduces our novel Hybrid Polynomial Zonotopes representation. Section~\ref{sec-propPWNA} presents algorithms for reachability analysis using the proposed representation. Section~\ref{sec-numericalEx} evaluates our approach on benchmark systems, and Section~\ref{sec-conclusion} concludes the paper with directions for future work.
\section{Notation and preliminaries}\label{sec-prelim}
Matrices are denoted by \emph{uppercase Roman letters} (e.g.\ $G\in\mathbb{R}^{n\times n_g}$), whereas \emph{calligraphic capitals} represent sets (e.g.\ $\mathcal{Z}\subset\mathbb{R}^{n}$).  
Vectors and scalars appear as \emph{lower-case Roman letters} (e.g.\ $b\in\mathbb{R}^{n_c}$).
The $n$-dimensional unit hyper-cube is
$
  \mathcal{B}_{\infty}^{n}
  \;=\;
  \bigl\{x\in\mathbb{R}^{n}\,\bigl|\,\|x\|_{\infty}\le 1\bigr\},
$
and the set of all $n$-dimensional binary vectors is $\{-1,1\}^{n}$.
For a finite set $\mathcal{T}$, its cardinality is denoted $|\mathcal{T}|$; for example, $|\{-1,1\}^{3}| = 8$.
The concatenation of two column vectors is denoted by
$
  (\xi_{1}\,\xi_{2})
  \;=\;
  \bigl[\xi_{1}^{\top}\;\xi_{2}^{\top}\bigr]^{\top}.
$
Bold symbols $\mathbf{1}$ and $\mathbf{0}$ denote the all-ones and all-zeros matrices, respectively, while $\mathbf{I}$ denotes the identity matrix (dimension clear from context).  
A superscript ${}^\top$ indicates the transpose.
For a matrix $A$, entry $(i,j)$ is written $A_{(i,j)}$ and column $j$ is $A_{(:,j)}$.  
Indexed objects carry an integer subscript; for instance, $A_{n}$ is the $n$-th matrix in the family $\{A_{n}\}_{n\in\mathbb{N}}$.
Composite dimension symbols use superscripts: for example, $n_{g_1}$ refers to $n_{g}$ associated with the set $\mathcal{Z}_{1}$.  
Accordingly, element $(i,j)$ of $A_{n}$ is $A_{n,(i,j)}$, and entry $j$ of $\xi_{n}$ is $\xi_{n,j}$.  
Finally, for each $n\in\mathbb{N}$, let $A_{c_n}$ denote the $n$-th matrix in the family $\{A_{c_k}\}_{k\ge 1}$; we write $A_{c_n,(i,j)}$ for the $(i,j)$-entry of $A_{c_n}$.

\subsection{Set representations}

Before we introduce the hybrid polynomial
zonotope, we give a brief introduction about constrained zonotope, constrained polynomial zonotope, hybrid zonotope, and their representation.
\begin{definition}[Constrained Zonotope~\cite{scott2016constrained}]\label{def-conZono} 
   The set $\mathcal{Z}_c\subset{\rm I\!R}^n$ is a constrained zonotope if there exist $G_c\in{\rm I\!R}^{n\times n_g}$, $c\in{\rm I\!R}^{n}$, $A_c\in{\rm I\!R}^{n_c\times n_g}$, and $b\in{\rm I\!R}^{n_c}$ such that
    \begin{equation}\label{def-eqn-conZono}
        \mathcal{Z}_c=\left\{G_c\xi+c\:\middle|\:\|\xi\|_\infty\leq1, A_c\xi=b\right\} \: .
    \end{equation}
\end{definition}
The constrained zonotope is given in Constrained Generator-representation (CG-rep), and the shorthand notation of $\mathcal{Z}_c=\langle c,G_c,A_c,b\rangle_\mathcal{CZ}$ is used to denote the set given by \eqref{def-eqn-conZono}. Through the addition of the linear equality constraints $A_c\xi=b$ to the projected unit hypercube, the affine image of the constrained space of factors is no longer restricted to be symmetric.
\begin{definition}[Constrained Polynomial Zonotope~\cite{kochdumper2023}]
	 Given a constant offset $c \in R^n$, a generator matrix $G_c \in R^{n \times h}$, an exponent matrix $E \in \mathbb{N}_{0}^{p \times h}$, a constraint generator matrix $A_c \in R^{m \times q}$, a constraint vector $b \in R^m$, and a constraint exponent matrix $R \in \mathbb{N}_{0}^{p \times q}$, a constrained polynomial zonotope is defined as
	\begin{align*}
		&\mathcal{CPZ} = \bigg \{ c + \sum_{i=1}^{h} \bigg( \prod_{k=1}^p \xi_k^{E_{(k,i)}} \bigg) G_{c,(\cdot,i)} ~ \bigg | ~\\& \sum_{i=1}^{q} \bigg( \prod_{k=1}^p \xi_k^{R_{(k,i)}} \bigg) A_{c,(\cdot,i)} = b, ~ \xi_k \in [-1,1]   \bigg \}.
	\end{align*}
	The scalars $\xi_k$ are called factors, where the number of factors is $p$, the number of generators $G_{c,(\cdot,i)}$ is $h$, the number of constraints is $m$, and the number of constraint generators $A_{c,(\cdot,i)}$ is $q$. The order $\rho = \frac{h+q}{n}$ estimates the complexity of a constrained polynomial zonotope. We use the shorthand $\mathcal{CPZ} = \langle c,G_c,E,A_c,b,R \rangle_{\mathcal{CPZ}}$.
	\label{def:CPZ}
\end{definition}
\begin{definition}[Hybrid Zonotope~\cite{bird2023hybrid}]\label{def-hybridZono}
    The set $\mathcal{Z}_h\subset{\rm I\!R}^n$ is a hybrid zonotope if there exist $G_c\in{\rm I\!R}^{n\times n_{g}}$, $G_b\in{\rm I\!R}^{n\times n_{b}}$, $c\in{\rm I\!R}^{n}$, $A_c\in{\rm I\!R}^{n_{c}\times n_{g}}$, $A_b\in{\rm I\!R}^{n_{c}\times n_{b}}$, and $b\in{\rm I\!R}^{n_c}$ such that
    \begin{equation}\label{def-eqn-hybridZono}
        \mathcal{Z}_h = \left\{ \left[G_c \: G_b\right]\left[\begin{smallmatrix}\xi_c \\ \xi_b \end{smallmatrix}\right]  + c\: \middle| \begin{matrix} \left[\begin{smallmatrix}\xi_c \\ \xi_b \end{smallmatrix}\right]\in \mathcal{B}_\infty^{n_{g}} \times \{-1,1\}^{n_{b}}, \\ \left[A_c \: A_b\right]\left[\begin{smallmatrix}\xi^c \\ \xi^b \end{smallmatrix}\right] = b \end{matrix} \right\}\:.
\end{equation}
\end{definition}
The hybrid zonotope is given in \textit{Hybrid Constrained Generator-representation} (HCG-rep), and the shorthand notation of $\mathcal{Z}_h=\langle G_c,G_b,c,A_c,A_b,b\rangle_\mathcal{HZ} \subset{\rm I\!R}^n$ is used to denote the set given by \eqref{def-eqn-hybridZono}. When $n_{b}=0$, the hybrid zonotope set representation is equivalent to the constrained zonotope given by Definition \ref{def-conZono}. When $n_{b}\not=0$, the vector of binary factors may take on values from the discrete set ${\{-1,1\}^{n_{b}}}$ containing $2^{n_{b}}$ elements. The degrees of freedom of a hybrid zonotope is a function of both the number of continuous and binary factors with order $o_d=(n_g+n_b-n_c)/n$.
Given that $\|\xi_b\|_{\infty}=1\;\forall\:\xi_b\in\{-1,1\}^{n_{b}}$, the hybrid zonotope is a more general class than the constrained zonotope set representation.

\section{Hybrid Polynomial Zonotopes Representation}\label{sec-HPZ}
This section introduces the definition of hybrid polynomial zonotopes as an extension of the constrained polynomial zonotope through the addition of a vector of binary factors.

\begin{definition}[Hybrid Polynomial Zonotope] 
\label{def-hybridpolyZono}
    The set $\mathcal{Z}\subset{\rm I\!R}^n$ is a hybrid polynomial zonotope if there exist $G_c\in{\rm I\!R}^{n\times n_{g}}$, $G_b\in{\rm I\!R}^{n\times n_{b}}$, $c\in{\rm I\!R}^{n}$, $E\in{\rm I\!R}^{n_e\times n_g}$, $A_c\in{\rm I\!R}^{n_{c}\times n_{g}}$, $A_b\in{\rm I\!R}^{n_{c}\times n_{b}}$, $b\in{\rm I\!R}^{n_c}$,   and $R\in{\rm I\!R}^{n_e\times n_g}$such that 
    \begin{align}
		&\mathcal{HPZ} = \bigg \{ c +G_b\xi_b+ \sum_{i=1}^{n_g} \bigg( \prod_{k=1}^{n_e} \xi_{c,k}^{E_{(k,i)}} \bigg) G_{c,(\cdot,i)} ~
        \bigg |  A_b\xi_b+\nonumber\\&\sum_{i=1}^{n_g} \bigg( \prod_{k=1}^{n_e} \xi_{c,k}^{R_{(k,i)}} \bigg) A_{c,(\cdot,i)} = b, \left[\begin{smallmatrix}\xi_c \\ \xi_b \end{smallmatrix}\right]\in \mathcal{B}_\infty^{n_{e}} \times \{-1,1\}^{n_{b}}  \bigg \}.\label{defhybridZono}
	\end{align}
\end{definition}
The hybrid polynomial zonotope is given in \textit{Hybrid Polynomial Constrained Generator-representation} (HPCG-rep), and the shorthand notation of $\mathcal{Z}=\langle c,G_c,G_b,E,A_c,A_b,b,R\rangle_\mathcal{HPZ}\subset{\rm I\!R}^n$ is used to denote the set given by \eqref{defhybridZono}. 
\begin{remark}
When $E=I$, $R=I$, the HPZ set representation $\mathcal{Z}=\langle c,G_c,G_b,E,A_c,A_b,b,R\rangle_\mathcal{HPZ}$ is equivalent to the hybrid zonotope given by Definition \ref{def-hybridZono}. When $n_{b}=0$, the HPZ set representation $\mathcal{Z}=\langle c,G_c,\emptyset,E,A_c,\emptyset,b,R\rangle_\mathcal{HPZ}$ is equivalent to the constrained 
polynomial zonotope given by Definition \ref{def:CPZ}. Furthurmore, the HPZ set representation $\mathcal{Z}=\langle c,G_c,\emptyset,I,A_c,\emptyset,b,I\rangle_\mathcal{HPZ}$ and $\mathcal{Z}=\langle c,G_c,\emptyset,I,\emptyset,\emptyset,\emptyset,\emptyset\rangle_\mathcal{HPZ}$ is equivalent to the constrained zonotope given by Definition \ref{def-conZono} and zonotope.
\end{remark}

\begin{example}\label{exa-hybridpolyZono}
Let the set $\mathcal{CPZ}_1=\left\langle c_1,G_1,E_1,A_1,b_1,R_1\right\rangle_\mathcal{CPZ}\subset{\rm I\!R}^2$ be the constrained polynomial zonotope shown in Fig.~\ref{fig-Example1}, where
\begin{align*}
\mathcal{CPZ}_1=&\langle\begin{bmatrix}0\\0\end{bmatrix},\begin{bmatrix}1 & 0 & 1.5&0.5\\0&1&2&-2\end{bmatrix},\begin{bmatrix}1 & 0 & 0&0\\0&1&0&2\\0&0&1&1\end{bmatrix},\\
  &\;\begin{bmatrix}1&2&0.5\end{bmatrix},1,\begin{bmatrix}1 & 0 & 0\\0&1&0\\0&0&3\end{bmatrix}\rangle_{\mathcal{CPZ}}.
\end{align*}
Next, define the hybrid zonotope $\mathcal{HZ}=\left\langle G_1,G_{b_1},c_1,\emptyset,\emptyset,\emptyset\right\rangle_\mathcal{HZ}$.  This object compactly represents a non-convex set obtained as the union of an exponential number of convex subsets as shown in Fig.~\ref{fig-Example1}, while requiring only a linear number of continuous and discrete variables. 
Define a HPZ with continuous generators $G_c=G_1$, binary generators $G_{b,1}=\begin{bmatrix}3 & 1 & 4\\1&3&-4\end{bmatrix}$, exponent matrix $E=E_1$, continuous constraint generators $A_c=A_1$, binary constraint generators $A_{b,1}=\emptyset$, exponent constraint matrix $R=R_1$, constraint vector $b=b_1$, and center $c=c_1$, giving $\mathcal{HPZ}_{1}=\left\langle c_1,G_1,G_{b_1},E_1,A_1,A_{b_1},b_1,R_1\right\rangle_\mathcal{HPZ}$. By adding $n_b=3$ binary factors, $\mathcal{HPZ}_{1}$ equals $2^{n_b}=8$ translated copies of $\mathcal{CPZ}_1$ whose centers are shifted by $G_{b_1}\xi_b$ for every $\xi_b\in\{-1,1\}^3$, as depicted in Fig.~\ref{fig-Example1}. Including the continuous and binary factors in the equality constraints by introducing $A_{b_2}=\begin{bmatrix}1.5 & 1.5 & 1.5\end{bmatrix}$ yields $\mathcal{HPZ}_{2}=\left\langle c_1,G_1,G_{b_1},E_1,A_1,A_{b_2},b_1,R_1\right\rangle_\mathcal{HPZ}$, also illustrated in Fig.~\ref{fig-Example1}; unlike $\mathcal{HPZ}_{1}$, these eight sets are not identical because the nonlinear equality constraints on the continuous factors shift with each binary assignment, and one such assignment renders the constraints infeasible, leaving six non-empty constrained polynomial zonotopes in this example. By contrast, an HPZ can encode, with similar compactness, the non-convex union of exponentially many non-convex sets.
\end{example}

\begin{figure}[htbp]
    \centering
    \includegraphics[width=0.48\textwidth]{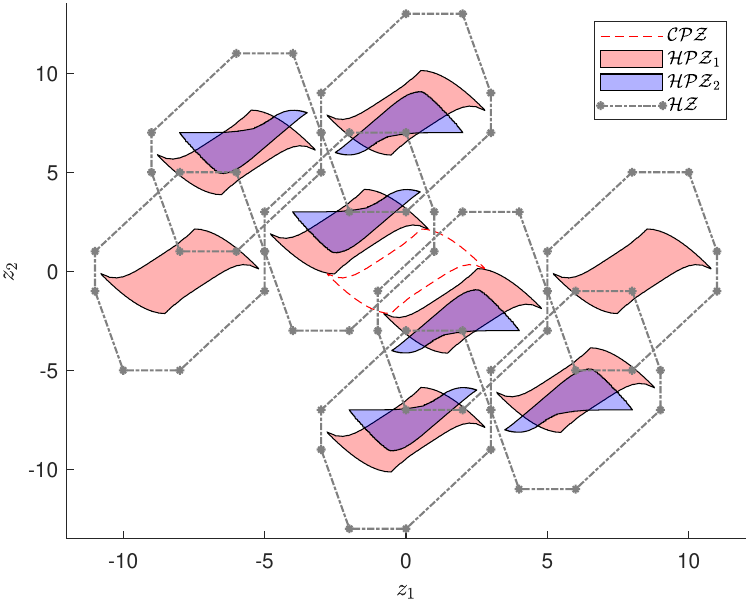}
    \caption{Hybrid polynomial zonotopes compared with Constrained Polynomial Zonotopes and Hybrid Zonotopes given in Ex. \ref{exa-hybridpolyZono}.}
    \label{fig-Example1}
\end{figure}
\subsection{Basic set operations with hybrid polynomial zonotopes}

The identities for linear mappings, Minkowski sums, generalized intersections, and halfspace intersections of constrained polynomial zonotopes may be extended to HPZs as follows. Beyond these basic set operations, the closure of HPZ under unions has been proven in this section.

\begin{proposition}[Linear Map] Given HPZ $\mathcal{Z}=\langle c,G_c,G_b,E,A_c,A_b,b,R\rangle_\mathcal{HPZ}\in{\rm I\!R}^{n}$ and a matrix $M\in{\rm I\!R}^{m\times n}$, the following identities hold:
\begin{align} \label{linearmap}
    M \otimes \mathcal{Z}=\langle Mc,MG_c,MG_b,E,A_c,A_b,b,R\rangle_\mathcal{HPZ}.
\end{align}
\end{proposition}
\begin{proof}
Our proof follows the lines of \cite{bird2023hybrid}, with additional refinements for the non-convex case. Let $\mathcal{Z}_M$ denote the HPZ given by the right-hand side of {\eqref{linearmap}}. For any point $z\in \mathcal{Z}$ there exists some $\xi_m=(\xi_{c_m},\xi_{b_m})\in\mathcal{B}_{\infty}^{n_{e_m}}\times\{-1,1\}^{n_{b_m}}$ such that multiplying both sides of $z$ by $M$ gives 
\begin{align}\label{Map1}
&\sum_{i=1}^{n_{g_m}} \bigg( \prod_{k=1}^{n_{e_m}} {\xi_{c_m,k}}^{R_{(k,i)}} \bigg) A_{c,(\cdot,i)}+A_{b} \xi_{b_z} ={b}, 
\end{align}
and 
\begin{align}\label{Map2}
    &       Mz\!=\!Mc+MG_{b} \xi_{b_m}+ M\sum_{i=1}^{n_{g_m}} \bigg( \prod_{k=1}^{n_{e_m}} {\xi_{c_m,k}}^{E_{(k,i)}} \bigg) G_{c,(\cdot,i)},
 \end{align}
 thus $M\mathcal{Z}\subseteq \mathcal{Z}_M$. Conversely, for any point $z_m\in \mathcal{Z}_M$ there exists some $\xi_m\in\mathcal{B}_{\infty}^{n_{e_m}}\times\{-1,1\}^{n_{b_m}}$ such that \eqref{Map1} and \eqref{Map2} are satisfied. Thus there exists some $z\in \mathcal{Z}$ such that $Mz=z_m$. Therefore $\mathcal{Z}_M\subseteq M\mathcal{Z}$ and $M\mathcal{Z}=\mathcal{Z}_M$.
\hfill
\end{proof}

\begin{proposition}[Minkowski Sum] Given HPZ $\mathcal{Z}_1=\langle c_1,G_{c_1},G_{b_1},E_1,A_{c_1},A_{b_1},b_1,R_1\rangle_\mathcal{HPZ}$ and $\mathcal{Z}_2=\langle c_2,G_{c_2},G_{b_2},E_2,A_{c_2},A_{b_2},b_2,R_2\rangle_\mathcal{HPZ}$, their Minkowski sum is:
\begin{align}
   & \mathcal{Z}_1 \oplus \mathcal{Z}_2=\bigg \langle c_1 + c_2,\begin{bmatrix} G_{b_1} & G_{b_2} \end{bmatrix}, \begin{bmatrix} G_{c_1} & G_{c_2} \end{bmatrix}, \begin{bmatrix} E_1 & \mathbf{0} \\ \mathbf{0} & E_2 \end{bmatrix}, \nonumber\\&~~~~~~\begin{bmatrix} A_{b_1} & \mathbf{0} \\ \mathbf{0} & A_{b_2} \end{bmatrix}, \begin{bmatrix} A_{c_1} & \mathbf{0} \\ \mathbf{0} & A_{c_2} \end{bmatrix},\begin{bmatrix} b_1 \\ b_2 \end{bmatrix}, \begin{bmatrix} R_1 & \mathbf{0} \\ \mathbf{0} & R_2 \end{bmatrix}  \bigg \rangle_\mathcal{HPZ} \label{minksum}.
\end{align}
\end{proposition}
\begin{proof}
Our proof follows the lines of \cite{bird2023hybrid}, with additional refinements for the non-convex case. Let $\mathcal{Z}_S$ denote the HPZ given by the right-hand side of \eqref{minksum}. For any $z_1\in \mathcal{Z}_1$ there exists some $\xi_1\in\mathcal{B}_{\infty}^{n_{e_1}}\times\{-1,1\}^{n_{b_1}}$ such that
\begin{align}\label{constraintssum}
&\sum_{i=1}^{n_{g_m}} \bigg( \prod_{k=1}^{n_{e_m}} {\xi_{c_m,k}}^{R_{m,(k,i)}} \bigg) {A}_{c_m,(\cdot,i)}+A_{b_m} \xi_{b_m} ={b}_{m}, 
\end{align}
and 
\begin{align}\label{zonotopesum}
    & z_m=c_m+G_{b_m} \xi_{b_m}+ \sum_{i=1}^{n_{g_m}} \bigg( \prod_{k=1}^{n_{e_m}} {\xi_{c_m,k}}^{E_{m,(k,i)}} \bigg) G_{c_m,(\cdot,i)}
 \end{align} with $m=1$. Similarly for any $z_2\in \mathcal{Z}_2$ there exists some $\xi_2\in\mathcal{B}_{\infty}^{n_{e_2}}\times\{-1,1\}^{n_{b_2}}$ such that 
\eqref{constraintssum} and \eqref{zonotopesum} with $m=2$. Let $\xi_{c_s}=(\xi_{c_1}~\xi_{c_2})$ and $\xi_{b_s}=(\xi_{b_1}~\xi_{b_2})$. Then $\xi_s\in\mathcal{B}_{\infty}^{n_{e_1}+n_{e_2}}\times\{-1,1\}^{n_{b_1}+n_{b_2}}$ and adding $z_1$ and $z_2$ together gives
\begin{small}
	\begin{align}
		&\mathcal{Z}_1 \oplus \mathcal{Z}_2 {=} \big\{ z_1 + z_2  \big| z_1 \in \mathcal{Z}_1,~z_2 \in \mathcal{Z}_2 \big\}  {=} \nonumber\\
		& \bigg \{ c_1 + c_2 + G_{b_1}\xi_{b_1}+ G_{b_2}\xi_{b_2}+ \sum_{i=1}^{n_{g_1}} \bigg( \prod_{k=1}^{n_{e_1}} {\xi_{c_1,k}}^{E_{1,(k,i)}} \bigg) G_{1,(\cdot,i)} + \nonumber\\& \sum_{i=1}^{n_{g_2}} \bigg( \prod_{k=1}^{n_{e_2}} {\xi_{c_2,k}}^{E_{2,(k,i)}} \bigg) G_{2,(\cdot,i)}  \bigg | \sum_{i=1}^{n_{g_1}} \bigg( \prod_{k=1}^{n_{e_1}} {\xi_{c_1,k}}^{R_{1,(k,i)}} \bigg) A_{c_1,(\cdot,i)} \nonumber\\&+A_{b_1}\xi_1 = b_1, ~ \sum_{i=1}^{n_{g_2}} \bigg( \prod_{k=1}^{n_{e_2}} {\xi_{c_2,k}}^{R_{2,(k,i)}} \bigg) A_{c_2,(\cdot,i)}+A_{b_2}\xi_2 = b_2,  \nonumber\\& ~~~~~~~\left[\begin{smallmatrix}\xi_{c_s} \\ \xi_{b_s} \end{smallmatrix}\right]\in \mathcal{B}_\infty^{n_{e_1}+n_{e_2}} \times \{-1,1\}^{n_{b_1}+n_{b_2}}   \bigg \} {=} \nonumber\\& \bigg \langle c_1 + c_2,\begin{bmatrix} G_{b_1} & G_{b_2} \end{bmatrix}, \begin{bmatrix} G_{c_1} & G_{c_2} \end{bmatrix}, \begin{bmatrix} E_1 & \mathbf{0} \\ \mathbf{0} & E_2 \end{bmatrix}, \begin{bmatrix} A_{b_1} & \mathbf{0} \\ \mathbf{0} & A_{b_2} \end{bmatrix},  \nonumber\\& ~~~~~~~~~~~~~~\begin{bmatrix} A_{c_1} & \mathbf{0} \\ \mathbf{0} & A_{c_2} \end{bmatrix},\begin{bmatrix} b_1 \\ b_2 \end{bmatrix}, \begin{bmatrix} R_1 & \mathbf{0} \\ \mathbf{0} & R_2 \end{bmatrix}  \bigg \rangle_\mathcal{HPZ},
        \label{minkSum-zplusw}
    \end{align}
\end{small}
thus $z_1+z_2\in\mathcal{Z}_S$ and $\mathcal{Z}_1\oplus\mathcal{Z}_2\subseteq \mathcal{Z}_S$. Conversely, for any $z_s\in\mathcal{Z}_S$ there exists some $\xi_s\in\mathcal{B}_{\infty}^{n_{e_1}+n_{e_2}}\times\{-1,1\}^{n_{b_1}+n_{b_2}}$ such that \eqref{constraintssum} holds with $m=1,2$ and $z_s=z_1+z_2$ as defined by \eqref{minkSum-zplusw}. Letting $\xi_{c_s}=(\xi_{c_1}~\xi_{c_2})$ and $\xi_{b_s}=(\xi_{b_1}~\xi_{b_2})$ gives $z_s\in\mathcal{Z}_1\oplus\mathcal{Z}_2$ and $\mathcal{Z}_S\subseteq \mathcal{Z}_1\oplus\mathcal{Z}_2$, therefore $\mathcal{Z}_1\oplus\mathcal{Z}_2=\mathcal{Z}_S$, which concludes the proof.
 \hfill 
\end{proof}

\begin{proposition}[Generalized Intersection] Given HPZ $\mathcal{Z}_1=\langle c_1,G_{c_1},G_{b_1},E_1,A_{c_1},A_{b_1},b_1,R_1\rangle_\mathcal{HPZ}$ and $\mathcal{Z}_2=\langle c_2,G_{c_2},G_{b_2},E_2,A_{c_2},A_{b_2},b_2,R_2\rangle_\mathcal{HPZ}$, their intersection is:
\begin{align}
			&\mathcal{Z}_1 \cap \mathcal{Z}_2 = \bigg \langle c_1, \begin{bmatrix}G_{c_1} &\mathbf{0}\end{bmatrix}, \begin{bmatrix}G_{b_1} &\mathbf{0}\end{bmatrix}, \begin{bmatrix} E_1 &\mathbf{0}\\ \mathbf{0}&I \end{bmatrix},\nonumber\\& \begin{bmatrix} A_{c_1} & \mathbf{0} & \mathbf{0} & \mathbf{0} \\ \mathbf{0} & A_{c_2} & \mathbf{0} & \mathbf{0}  \\ \mathbf{0} & \mathbf{0} & G_{c_1} & -G_{c_2} \end{bmatrix}, \begin{bmatrix} A_{b_1} & \mathbf{0} & \mathbf{0} & \mathbf{0} \\ \mathbf{0} & A_{b_2} & \mathbf{0} & \mathbf{0}  \\ \mathbf{0} & \mathbf{0} & G_{b_1} & -G_{b_2} \end{bmatrix},\nonumber\\& \begin{bmatrix} b_1 \\ b_2 \\ c_2 - c_1 \end{bmatrix}, \begin{bmatrix} R_1 & \mathbf{0} & E_1 & \mathbf{0} \\ \mathbf{0} & R_2 & \mathbf{0} & E_2 \end{bmatrix}  \bigg \rangle_{\mathcal{HPZ}}.\label{genintersection}
\end{align}
\end{proposition}
\begin{proof}
Our proof follows the lines of \cite{bird2023hybrid}, with additional refinements for the non-convex case. Let $\mathcal{Z}_G$ denote the HPZ given by the right-hand side of {\eqref{genintersection}}. For any $z_g\in\mathcal{Z}_G$ there exists some $\xi_g\in\mathcal{B}_{\infty}^{n_{e_1}+n_{e_2}}\times\{-1,1\}^{n_{b_1}+n_{b_2}}$ such that
\begin{align} 
&c_1  + G_{b_1}\xi_{b_1} + \sum_{i=1}^{n_{g_1}} \bigg( \prod_{k=1}^{n_{e_1}} {\xi_{c_1,k}}^{E_{1,(k,i)}} \bigg) G_{1,(\cdot,i)} =  \nonumber\\&c_2+ G_{b_2}\xi_{b_2} + \sum_{i=1}^{n_{g_2}} \bigg( \prod_{k=1}^{n_{e_2}} {\xi_{c_2,k}}^{E_{2,(k,i)}} \bigg) G_{2,(\cdot,i)}, \label{gencon1}
\end{align}
\begin{align} 
    &\sum_{i=1}^{n_{g_1}} \bigg( \prod_{k=1}^{n_{e_1}} {\xi_{c_1,k}}^{R_{1,(k,i)}} \bigg) A_{c_1,(\cdot,i)}+A_{b_1}\xi_1 = b_1,\label{gencon2}
\end{align}
\begin{align} 
& \sum_{i=1}^{n_{g_2}} \bigg( \prod_{k=1}^{n_{e_2}} {\xi_{c_2,k}}^{R_{2,(k,i)}} \bigg) A_{c_2,(\cdot,i)}+A_{b_2}\xi_2 = b_2 \label{gencon3}
\end{align}
and
\begin{align} \label{gencres}
z_g=&c_1  + G_{b_1}\xi_{b_1}+ \mathbf{0}\xi_{b_2}+ \nonumber\\&\sum_{i=1}^{n_{g_1}} \bigg( \prod_{k=1}^{n_{e_1}} {\xi_{c_1,k}}^{E_{1,(k,i)}} \bigg) G_{1,(\cdot,i)}+\sum_{i=1}^{n_{g_2}} \mathbf{0} G_{1,(\cdot,i)}
\end{align}
Letting $\xi_{c_g}=(\xi_{c_1}~\xi_{c_2})$ and $\xi_{b_g}=(\xi_{b_1}~\xi_{b_2})$ gives
\begin{align}\label{genres1}
&z_g=c_1  + G_{b_1}\xi_{b_1}+\sum_{i=1}^{n_{g_1}} \bigg( \prod_{k=1}^{n_{e_1}} {\xi_{c_1,k}}^{E_{1,(k,i)}} \bigg) G_{1,(\cdot,i)}
\end{align}
and the satisfaction of constraint \eqref{gencon2}, thus $z_g\in \mathcal{Z}_1$. From the equality constraints \eqref{gencon1} and \eqref{gencon3}, we have \begin{align} \label{genres2}
&z_g=c_2  + G_{b_2}\xi_{b_2}+\sum_{i=1}^{n_{g_2}} \bigg( \prod_{k=1}^{n_{e_2}} \xi_{c,k}^{E_{2(k,i)}} \bigg) G_{2(\cdot,i)}
\end{align}
and the satisfaction of constraint \eqref{gencon3}, which gives $z_g\in \mathcal{Z}_2$. Therefore $z_g\in\mathcal{Z}_1\cap \mathcal{Z}_2$ and $\mathcal{Z}_G\subseteq \mathcal{Z}_1\cap \mathcal{Z}_2$. Conversely, for any $z_g\in \mathcal{Z}_1\cap \mathcal{Z}_2$ there exists some $\xi_1\in\mathcal{B}_{\infty}^{n_{e_1}}\times\{-1,1\}^{n_{b_1}}$ such that \eqref{genres1} and \eqref{gencon2} are satisfied. Furthermore, there exists $\xi_2\in\mathcal{B}_{\infty}^{n_{e_2}}\times\{-1,1\}^{n_{b_2}}$ such that \eqref{genres2} and \eqref{gencon3} are satisfied. Letting $\xi_{c_g}=(\xi_{c_1}~\xi_{c_2})$ and $\xi_{b_g}=(\xi_{b_1}~\xi_{b_2})$ implies that $\xi_g\in\mathcal{B}_{\infty}^{n_{e_1}+n_{e_2}}\times\{-1,1\}^{n_{b_1}+n_{b_2}}$ satisfies {\eqref{gencon1}}, {\eqref{gencon2}}, and {\eqref{gencon3}}, and $z_1=z_g$ in \eqref{gencres}. Therefore, $z_g\in\mathcal{Z}_G$, $\mathcal{Z}_1\cap \mathcal{Z}_2 \subseteq \mathcal{Z}_G$, and $\mathcal{Z}_1\cap \mathcal{Z}_2=\mathcal{Z}_G$.    \hfill
\end{proof}
\begin{proposition}[Halfspace Intersection]\label{halfspaceintersection} Given HPZ $\mathcal{Z}=\langle c,G_c,G_b,E,A_c,A_b,b,R\rangle_\mathcal{HPZ}\in{\rm I\!R}^{n}$, $M\in{\rm I\!R}^{m\times n}$ and $\mathcal{H}^{-}=\{x\in{\rm I\!R}^m~|~l^Tx\leq \rho\}$, their intersection is:
\begin{align}
\mathcal{Z}_h&\cap_M\mathcal{H}^{-}=\langle c,\begin{bmatrix}G_c & \mathbf{0} \end{bmatrix},G_b, \begin{bmatrix}E &\mathbf{0} \\ \mathbf{0}&I \end{bmatrix},\nonumber\\&\begin{bmatrix}A_c &\mathbf{0}&\mathbf{0}\\ \mathbf{0} & l^TMG_c & -1/2(\rho-l_m)\end{bmatrix}, 
      \begin{bmatrix}A_b\\h^TMG_b \end{bmatrix},\nonumber\\&\begin{bmatrix}b\\1/2(\rho+l_m)-l^TMc  \end{bmatrix},   \begin{bmatrix}
\begin{bmatrix}R \\ \mathbf{0}\end{bmatrix} & \begin{bmatrix} E & \mathbf{0} \\ \mathbf{0} & I \end{bmatrix}
\end{bmatrix} \rangle_\mathcal{HPZ},\label{halfintersection}
\end{align}
where
\begin{small}
\begin{align}
      & l_m=h^TMc-\sum_{i=1}^{n_{g}}\vert \bigg( \prod_{k=1}^{n_e} \xi_{c,k}^{E_{(k,i)}} \bigg) l^TMG_{c,(\cdot,i)}\vert -\sum_{i=1}^{n_{b}}\vert h^TMG_{b,(\cdot,i)}\vert.\label{lm}
\end{align}
\end{small}

\end{proposition}
\begin{proof}
Our proof follows the lines of \cite{bird2023hybrid}, with additional refinements for the non-convex case. Let $\mathcal{Z}_H$ denote the HPZ given by the right-hand side of {\eqref{halfintersection}}. For any $z_h\in\mathcal{Z}_H$ there exists some $\xi_h\in\mathcal{B}_{\infty}^{n_{e}+1}\times\{-1,1\}^{n_{b}}$ such that 
\begin{align}
&\sum_{i=1}^{n_g} \bigg( \prod_{k=1}^{n_e} {\xi_{c_h,k}}^{E_{(k,i)}} \bigg)l^T MG_{c,(\cdot,i)}-  \xi_{f} 1/2(f-l_m) \nonumber\\&~~~~~~~~~~~+l^TMG_b\xi_{h_b}= (1/2(f+l_m)-l^TMc),\label{halfcon1}\\
    &\sum_{i=1}^{n^g} \bigg( \prod_{k=1}^{n^e} {\xi_{c_h,k}}^{R_{(k,i)}} \bigg) A_{c,(\cdot,i)}+A_b\xi_{b_h} = b,\label{halfcon2}
\end{align}
and
\begin{align}
&z_h=c  + G_b\xi_b+ \sum_{i=1}^{n_g} \bigg( \prod_{k=1}^{n_e} \xi_{c,k}^{E_{(k,i)}} \bigg) G_{c,(\cdot,i)}+\xi_f\mathbf{0}. \label{halfres}
\end{align}
Let $\xi_{c_h}=(\xi_c~\xi_f)$ and $\xi_{b_h}=\xi_b$ for $\xi_c\in {\rm I\!R}^{n_{e}}$, $\xi_f\in {\rm I\!R}$, and $\xi_b\in\{-1,1\}^{n_{b}}$. Then we have $z_h\in\mathcal{Z}$. Next, for
\begin{align*}
l^TMz_h&=l^TMc  + l^TMG_b\xi_b+ \sum_{i=1}^{n_g} \bigg( \prod_{k=1}^{n_e} \xi_{c,k}^{E_{(k,i)}} \bigg) l^TMG_{c,(\cdot,i)}\\&
= \xi_{f} 1/2(\rho-l_m) + 1/2(\rho+l_m).
\end{align*}
From the definition of $l_m$ in \eqref{lm} and that $\|\xi_f\|_{\infty}\leq1$, it follows that
\begin{small}
\protect\begin{equation}
\begin{split}
    &l^T Mz_h\in\\
    &\left[h^TMc-\sum_{i=1}^{n_{g}}\vert \bigg( \prod_{k=1}^{n_e} \xi_{c,k}^{E_{(k,i)}} \bigg) l^TMG_{c,(\cdot,i)}\vert -\sum_{i=1}^{n_{b}}\vert h^TMG_{b,(\cdot,i)}\vert , \rho \right],
\end{split}
\end{equation}
\end{small}therefore $Mz_h\in\mathcal{H}^{-}$ and $\mathcal{Z}_H\subseteq\mathcal{Z}\cap_M\mathcal{H}^{-}$. Conversely, for any point $z_h\in \mathcal{Z}\cap_R\mathcal{H}^{-}$ there exists some $\xi\in\mathcal{B}_{\infty}^{n_{e}}\times\{-1,1\}^{n_{b}}$ such that
\begin{align*}
    &\sum_{i=1}^{n_g} \bigg( \prod_{k=1}^{n_e} \xi_{c,k}^{R_{(k,i)}} \bigg) A_{c,(\cdot,i)}+A_b\xi_b = b,
\end{align*}
\begin{align*}
&z_h=c  + G_b\xi_b+ \sum_{i=1}^{n_g} \bigg( \prod_{k=1}^{n_e} \xi_{c,k}^{E_{(k,i)}} \bigg) G_{c,(\cdot,i)},
\end{align*}
and $l^T M z_h\leq \rho$. Thus $l^TMz_h\in[l_m,f]$ for some $l_m\leq l^TMz_h$ for all $z_h\in\mathcal{Z}\cap_M\mathcal{H}^{-}$. Choose 
\begin{small}
\begin{align*}
    l_m\!=\!h^TMc\!-\!\sum_{i=1}^{n_{g}}\vert \bigg( \prod_{k=1}^{n_e} \xi_{c,k}^{E_{(k,i)}} \bigg) l^TMG_{c,(\cdot,i)}\vert\! -\!\sum_{i=1}^{n_{b}}\vert h^TMG_{b,(\cdot,i)}\vert 
\end{align*}
\end{small}and let
\begin{small}
\begin{align*}
    u_m\!=\!h^TMc\!+\!\sum_{i=1}^{n_{g}}\vert \bigg( \prod_{k=1}^{n_e} \xi_{c,k}^{E_{(k,i)}} \bigg) l^TMG_{c,(\cdot,i)}\vert \!+\!\sum_{i=1}^{n_{b}}\vert h^TMG_{b,(\cdot,i)}\vert 
\end{align*}
\end{small}then $h^TM\mathcal{Z}_h\subseteq[l_m,u_m]$. Let $\xi_{c_h}=(\xi_c~\xi_{f})$ and $\xi_{b_h}=\xi_b$.
The above then implies that $\xi_h\in\mathcal{B}_{\infty}^{n_{e}+1}\times\{-1,1\}^{n_{b}}$ satisfies {\eqref{halfcon1}} and \eqref{halfcon2}, and $z_h\in\mathcal{Z}_H$. Therefore $\mathcal{Z}\cap_M\mathcal{H}^{-}\subseteq\mathcal{Z}_H$ and $\mathcal{Z}\cap_M\mathcal{H}^{-}=\mathcal{Z}_H$.   \hfill
\end{proof}

\begin{definition}\label{union}(Union) Given HPZs $\mathcal{Z}_1=\langle c_1,G_{c_1},G_{b_1},E_1,A_{c_1},A_{b_1},b_1,R_1\rangle_\mathcal{HPZ}$ and $\mathcal{Z}_2=\langle c_2,G_{c_2},G_{b_2},E_2,A_{c_2},A_{b_2},b_2,R_2\rangle_\mathcal{HPZ}$, the Union is: 
\begin{align}
& \mathcal{Z}_1 \cup \mathcal{Z}_2 = \bigg \langle {\hat{c}}, {\left[
G_{c_1} G_{c_2}  \mathbf{0}
\right]}, {\left[
G_{b_1} G_{b_2}  \hat{G}_b
\right]},{\begin{bmatrix} E_1 & \mathbf{0} & \mathbf{0}\\  \mathbf{0} & E_2 &\mathbf{0}\\  \mathbf{0} & \mathbf{0} & I \end{bmatrix}}, \nonumber \\& {\left[\begin{array}{cccc}
A_{c_1} & \mathbf{0} & \mathbf{0} & \mathbf{0}\\
\mathbf{0} & A_{c_2} & \mathbf{0}& \mathbf{0} \\
\mathbf{0}&\mathbf{0}& A_{c_3} & I
\end{array}\right]},  {\left[\begin{array}{ccc}
A_{b_1} & \mathbf{0} & \hat{A}_{b_1} \\
\mathbf{0} & A_{b_2} & \hat{A}_{b_2} \\
& A_{b_3} &
\end{array}\right]},\nonumber\\&~~~~~~~~~~~~{\left[\begin{array}{l}
\hat{b}_1 \\
\hat{b}_2 \\
b_3
\end{array}\right]}, {\begin{bmatrix} R_1 & \mathbf{0} & \mathbf{0}& \mathbf{0}&\mathbf{0}\\  \mathbf{0} & R_2 &\mathbf{0}& \mathbf{0}&\mathbf{0}\\  \mathbf{0} & \mathbf{0} & E_1& \mathbf{0}&\mathbf{0}\\ \mathbf{0} &\mathbf{0} &\mathbf{0} &E_2&\mathbf{0}\\\mathbf{0}&\mathbf{0}&\mathbf{0}&\mathbf{0}&I\end{bmatrix}} \bigg \rangle_\mathcal{HPZ}, \label{unionhpz}
\end{align}
\begin{align*}
&b_3=\left[\begin{array}{cccccccc}
\frac{1}{2} \mathbf{1} & \frac{1}{2} \mathbf{1} & \frac{1}{2} \mathbf{1} & \frac{1}{2} \mathbf{1} & \mathbf{0} & \mathbf{1} & \mathbf{0} & \mathbf{1}
\end{array}\right]^T,\\
&A_{3_c}=\left[\begin{array}{cccccccc}
I & -I & \mathbf{0} & \mathbf{0} & \mathbf{0} & \mathbf{0} & \mathbf{0} & \mathbf{0} \\
\mathbf{0} & \mathbf{0} & I & -I & \mathbf{0} & \mathbf{0} & \mathbf{0} & \mathbf{0}
\end{array}\right]^T,\\
&\begin{small}A_{3_b}=\left[\begin{array}{cccccccc}
\mathbf{0} & \mathbf{0} & \mathbf{0} & \mathbf{0} & \frac{1}{2} I & -\frac{1}{2} I & \mathbf{0} & \mathbf{0} \\
\mathbf{0} & \mathbf{0} & \mathbf{0} & \mathbf{0} & \mathbf{0} & \mathbf{0} & \frac{1}{2} I & -\frac{1}{2} I \\
\frac{1}{2} \mathbf{1} & \frac{1}{2} \mathbf{1} & -\frac{1}{2} \mathbf{1} & -\frac{1}{2} \mathbf{1} & \frac{1}{2} \mathbf{1} & \frac{1}{2} \mathbf{1} & -\frac{1}{2} \mathbf{1} & -\frac{1}{2} \mathbf{1}
\end{array}\right]^T.
\end{small}
\end{align*}
where given the vectors $\hat{G}_b \in \mathbb{R}^n, \hat{c} \in \mathbb{R}^n, \hat{A}_{b_1} \in \mathbb{R}^{n_{c_1}}, \hat{b}_1 \in \mathbb{R}^{n_{c_1}}, \hat{A}_{b_2} \in \mathbb{R}^{n_{c_2}}$, and $\hat{b}_2 \in \mathbb{R}^{n_{c_2}}$ such that
\begin{align*}
  &  \left[\begin{array}{cc}
I & I \\
-I & I
\end{array}\right]\left[\begin{array}{c}
\hat{G}_b \\
\hat{c}
\end{array}\right]=\left[\begin{array}{c}
G_{b_2} \mathbf{1}+c_1 \\
G_{b_1} \mathbf{1}+c_2
\end{array}\right],\\&\left[\begin{array}{cc}
-I & I \\
I & I
\end{array}\right]\left[\begin{array}{c}
\hat{A}_{b_1} \\
\hat{b}_1
\end{array}\right]=\left[\begin{array}{c}
b_1 \\
-A_{b_1} \mathbf{1}
\end{array}\right],\\&\left[\begin{array}{cc}
-I & I \\
I & I
\end{array}\right]\left[\begin{array}{c}
\hat{A}_{b_2} \\
\hat{b}_2
\end{array}\right]=\left[\begin{array}{c}
-A_{b_2} \mathbf{1} \\
b_2
\end{array}\right] .
\end{align*}
\end{definition}
\begin{proof}
Our proof follows the lines of \cite{bird2021unions}, with additional refinements for the non-convex case. Let $\mathcal{Z}_U=\left\langle c_{u},G_{c_u}, G_{b_u},E_u,  A_{c_u}, A_{b_u}, b_{u},R_u\right\rangle_\mathcal{HPZ}$ denote the HPZ given by \eqref{unionhpz}. For any $z_u \in \mathcal{Z}_U$ there exists some $\xi_{c_u} \in$ $\mathcal{B}_{\infty}^{n_{e_u}}$ and $\xi_{b_u} \in\{-1,1\}^{n_{b_u}}$ such that
\begin{align*}
A_{b_u}\xi_{b_u}+\nonumber\sum_{i=1}^{n_{g_u}} \bigg( \prod_{k=1}^{n_{e_u}} {\xi_{c_u,k}}^{R_{u,(k,i)}} \bigg) A_{c_u,(\cdot,i)} = b_u
\end{align*}
 and
 \begin{align*}
     z_u=c_u +G_{b_u}\xi_{b_u}+ \sum_{i=1}^{n_{g_u}} \bigg( \prod_{k=1}^{n_{e_u}} {\xi_{c_u,k}}^{E_{u,(k,i)}} \bigg) G_{c_u,(\cdot,i)} ~
 \end{align*}
 Let $\xi_{c_u}=\left(\xi_{c_1} ~ \xi_{c_2} ~ \xi_{c_r}\right)$, where $\xi_{c_1} \in$ $\mathbb{R}^{n_{e_1}}, \xi_{c_2} \in \mathbb{R}^{n_{e_2}}$, and $\xi_{c_r} \in \mathbb{R}^{2\left(n_{e_1}+n_{e_2}+n_{b_1}+n_{b_2}\right)}$, and $\xi_{b_u}=$ $\left(\xi_{b_1}~ \xi_{b_2}~ \xi_{b_r}\right)$, where $\xi_{b_1} \in\{-1,1\}^{n_{ b_1}}, \xi_{b_2} \in\{-1,1\}^{n_{b_2}}$, and $\xi_{b_r} \in$ $\{-1,1\}$. Expanding the equality constraints constrained by $A_{c_3}$, $A_{b_3}$, and $b_3$ gives
 \begin{subequations} \label{factorcon}
\begin{align} \label{factorcon1}
&\xi_{c_1}  =1 / 2 \mathbf{1}-1 / 2 \xi_{b_r}-\xi_{c_r, 1}=-1 / 2 \mathbf{1}+1 / 2 \xi_{b_r}+\xi_{c_r, 2},  \\
\label{factorcon2}&\xi_{c_2}  =1 / 2 \mathbf{1}+1 / 2 \xi_{b_r}-\xi_{c_r, 3}=-1 / 2 \mathbf{1}-1 / 2 \xi_{b_r}+\xi_{c_r, 4},  \\
\label{factorcon3}&1 / 2 \xi_{b_1}  =-1 / 2 \xi_{b_r}-\xi_{c_r, 5}=-\mathbf{1}+1 / 2 \xi_{b_r}+\xi_{c_r, 6},  \\
\label{factorcon4}&1 / 2 \xi_{b_2}  =1 / 2 \xi_{b_r}-\xi_{c_r, 7}=-\mathbf{1}-1 / 2 \xi_{b_r}+\xi_{c_r, 8},
\end{align}
\end{subequations}
where $\xi_{c_r}=\left(\xi_{c_r, 1} ~ \cdots ~ \xi_{c_r, 8}\right)$. Letting $\xi_{b_r}=1$, \eqref{factorcon} reduces to
\begin{subequations}\label{conxi}
\begin{align} 
& \xi_{c_1}=-\xi_{c_r, 1}=\xi_{c_r, 2},  \\
& \xi_{c_2}=\mathbf{1}-\xi_{c_r, 3}=-\mathbf{1}+\xi_{c_r, 4},  \\
& \xi_{b_1}=-\mathbf{1}-2 \xi_{c_r, 5}=-\mathbf{1}+2 \xi_{c_r, 6},  \\
& \xi_{b_2}=\mathbf{1}-2 \xi_{c_r, 7}=-3 \mathbf{1}+2 \xi_{c_r, 8}. 
\end{align}
\end{subequations}
Given that $\left\|\xi_{c_r}\right\|_{\infty} \leq 1$, \eqref{factorcon2} and \eqref{factorcon4} are only satisfied for $\xi_{c_2}=\mathbf{0}$ and $\xi_{b_2}=-\mathbf{1}$ respectively, while \eqref{factorcon1} and \eqref{factorcon3} are satisfied for any $\left\|\xi_{c_1}\right\|_{\infty} \leq 1$ and $\xi_{b_1} \in\{-1,1\}^{n_{b_1}}$. Let $\xi_{c_u}=$ $\left(\xi_{c_1}~ \mathbf{0} ~\xi_{c_r}\right)$ and $\xi_{b_u}=\left(\xi_{b_1}~-\mathbf{1}~ 1\right)$. Expanding $z_u$ gives
 \begin{align}
    z_u=&\hat{c}+\hat{G}^{b}+G_{b_1} \xi_{b_1}c_1+\sum_{i=1}^{n_{g_1}} \bigg( \prod_{k=1}^{n_{e_1}} {\xi_{c_1,k}}^{E_{1,(k,i)}} \bigg) G_{c_1,(\cdot,i)} \nonumber\\&-G_{b_2} \mathbf{1} +\sum_{i=1}^{n_{g_2}} \mathbf{0} G_{c_2,(\cdot,i)} +\sum_{i=1}^{n_{g_r}} \bigg( \prod_{k=1}^{n_{e_r}} {\xi_{c_r,k}}^{I_{(k,i)}} \bigg) \mathbf{0}, \label{conres}
 \end{align}
and, after substituting $-G_{b_2} \mathbf{1}+\hat{G}_{b}+\hat{c}=c_{1}$, reduces to
 \begin{align}
    & z_u=c_1+G_{b_1} \xi_{b_1}+ \sum_{i=1}^{n_{g_1}} \bigg( \prod_{k=1}^{n_{e_1}} {\xi_{c_1,k}}^{E_{1,(k,i)}} \bigg) G_{c_1,(\cdot,i)}.
 \end{align}
Expanding the equality constraints constrained by the first two rows of the $A_{c_u}$, $A_{b_u}$ and $b_u$ results in
\begin{align}\label{coneq1}
&\sum_{i=1}^{n_{g_1}} \bigg( \prod_{k=1}^{n_{e_1}} {\xi_{c_1,k}}^{R_{1,(k,i)}} \bigg) A_{c_1,(\cdot,i)}+A_{b_1} \xi_{b_1}+\hat{A}_{b_1}  =\hat{b}_{1} ,
\end{align} 
\begin{align}\label{coneq2}
\sum_{i=1}^{n_{g_2}} \mathbf{0}A_{c_2,(\cdot,i)}-A_{b_2} \mathbf{1}+\hat{A}_{b_2}  =\hat{b}_{2},
\end{align}
which, after substituting $\hat{b}_{1}-\hat{A}_{b_1}=b_{1}$ and $\hat{b}_{2}-\hat{A}_{b_2}=$ $-A_{2}^{b} \mathbf{1}$, gives $-A_{b_2} \mathbf{1}=-A_{b_2} \mathbf{1}$ and 
\begin{align}\label{constraints}
&\sum_{i=1}^{n_{g_m}} \bigg( \prod_{k=1}^{n_{e_m}} {\xi_{c_1,k}}^{R_{m,(k,i)}} \bigg) {A}_{c_m,(\cdot,i)}+A_{b_m} \xi_{b_m} ={b}_{m}, 
\end{align} with $m=1$. Combining \eqref{conxi} results in $z_u \in \mathcal{Z}_{1}$ for $\xi_{b_r}=1$.
Similarly, let $\xi_{b_r}=-1$ we have
\begin{align}
    & z_u=c_2+G_{b_2} \xi_{b_2}+ \sum_{i=1}^{n_{g_2}} \bigg( \prod_{k=1}^{n_{e_2}} {\xi_{c_2,k}}^{E_{2,(k,i)}} \bigg) G_{c_2,(\cdot,i)}
 \end{align}
and constraint \eqref{constraints} is satisfied with $m=2$. Which means $z_u \in \mathcal{Z}_{2}$ for $\xi_{b_r}=-1$. Given that $\xi_{b_r} \in\{-1,1\}$ and that the choice of $z_u \in \mathcal{Z}_U$ is arbitrary, $\mathcal{Z}_U \subseteq \mathcal{Z}_{1} \cup \mathcal{Z}_{2}$.
Conversely, for any $z_m \in \mathcal{Z}_{m}$, where $m=1,2$, there exists some $\xi_{c_m} \in \mathcal{B}_{\infty}^{n_{e_m}}$ and $\xi_{b_m} \in\{-1,1\}^{n_{b_m}}$ such that constraints \eqref{constraints} are satisfied and
\begin{align}\label{zonotope}
    & z_m=c_m+G_{b_m} \xi_{b_m}+ \sum_{i=1}^{n_{g_m}} \bigg( \prod_{k=1}^{n_{e_m}} {\xi_{c_m,k}}^{E_{m,(k,i)}} \bigg) G_{c_m,(\cdot,i)}.
 \end{align}
For $m=1$, letting $\xi_{c_u}=\left(\xi_{c_1} ~\mathbf{0} ~\xi_{c_r}\right)$ and $\xi_{b_u}=\left(\xi_{b_1}~-\mathbf{1} ~1\right)$, \eqref{coneq1} and \eqref{coneq2} are satisfied and \eqref{conxi} implies that $\left\|\xi_{c_u}\right\|_{\infty} \leq 1$. Applying \eqref{conres} then gives $z_1\in \mathcal{Z}_U$. For $m=2$, letting $\xi_{c_u}=\left(\mathbf{0}~ \xi_{c_2} ~\xi_{c_r}\right)$ and $\xi_{b_u}=\left(-\mathbf{1} ~\xi_{b_2}~-1\right)$, similar to $m=1$ then gives $z_2 \in \mathcal{Z}_U$. Given that the choice of $z_1 \in \mathcal{Z}_{1}$ and $z_2 \in \mathcal{Z}_{2}$ is arbitrary, $\mathcal{Z}_{1} \cup \mathcal{Z}_{2} \subseteq \mathcal{Z}_U$ and therefore $\mathcal{Z}_{1} \cup \mathcal{Z}_{2}=\mathcal{Z}_U$.   \hfill
\end{proof}
\section{Forward Propagation of Hybrid Nonaffine Systems}\label{sec-propPWNA}

This section addresses the reachability analysis of discrete-time hybrid nonaffine systems with linear guards. 
\subsection{Problem Statement}
 We consider a system described by the following state equation:
\begin{equation}\label{eq:system}
x(k+1) = \begin{cases}
f_1(x(k), u(k)) & \text{if } \delta_1(k) = 1, \\
\vdots & \\
f_s(x(k), u(k)) & \text{if } \delta_s(k) = 1,
\end{cases}
\end{equation}

where $x(k) \in \mathbb{R}^{n_x}$ represents the system state, $u(k) \in \mathcal{U}_k \subset \mathbb{R}^{n_u}$ is the control input bounded by an input zonotope $\mathcal{U}_k$. The switching signal can be represented by binary variables $\delta_i(k) \in \{0,1\}$ indicating the active operating mode, which satisfy the exclusive-or condition:
\begin{equation}
\sum_{i=1}^s \delta_i(k) = 1.
\end{equation}

For each $j \in S$, the function $f_j: \mathbb{R}^{n_x} \times \mathbb{R}^{n_u} \rightarrow \mathbb{R}^{n_x}$ is an nonlinear nonaffine function. The initial state is bounded by the hybrid polynomial zonotope $\mathcal{X}_0=\langle c_{x_0},G_{c_{x_0}},G_{b_{x_0}},E_{x_0},A_{c_{x_0}},A_{b_{x_0}},b_{x_0},R\rangle_\mathcal{HPZ} \subset \mathbb{R}^{n_x}$.

For the system to be well-posed, the state space $\mathcal{C}$ must be partitioned into $s$ disjoint regions satisfying:
\begin{equation}
\mathcal{C}_i \cap \mathcal{C}_j = \emptyset, \quad \forall i \neq j, \quad \text{ and } \quad \bigcup_{i=1}^s \mathcal{C}_i = \mathcal{C}.
\end{equation}

Each region $\mathcal{C}_i$ is defined as a polyhedral set given by linear inequalities:
\begin{equation}
\mathcal{C}_i = \{x \in \mathbb{R}^{n_x} \mid L_i x \leq \rho_i\},
\end{equation}
where $L_i \in \mathbb{R}^{m_i \times n_x}$ and $\rho_i \in \mathbb{R}^{m_i}$. The binary variable $\delta_i(k) = 1$ if and only if $x(k) \in \mathcal{C}_i$.

The system exhibits both discrete switching behavior governed by the state-space partitioning and nonlinear continuous dynamics represented by the nonaffine functions $f_j$. This combination presents significant challenges for traditional reachability analysis methods, which typically rely on linear or affine approximations of system dynamics.



\subsection{Reachability Computation via Set Propagation}
We now present the algorithmic framework used in~\cite{xie2025data} for computing reachable sets of the hybrid nonaffine system~\eqref {eq:system}. The approach leverages HPZ representation to handle both discrete switching and nonlinear dynamics. At each time step $k$, the reachable set is represented as:
\begin{equation}
\mathcal{R}_k = \langle c, G_c, G_b, E, A_c, A_b, b, R \rangle_{\mathcal{HPZ}} \subset \mathbb{R}^{n_x}.
\end{equation}

The reachability computation proceeds iteratively through three fundamental operations:

\subsubsection{Mode-Dependent Set Partitioning}

Given $\mathcal{R}_k$, we compute the intersection with each state-space region $\mathcal{C}_i$:

\begin{equation}
\mathcal{R}_{k,i} = \mathcal{R}_k \cap \mathcal{C}_i = \text{halfspace}(\mathcal{R}_k, L_i, \rho_i), \quad \forall i \in S
\label{eq:intersection}
\end{equation}

where the halfspace operation \ref{halfspaceintersection} preserves the HPZ structure while enforcing the linear constraints $L_i x \leq \rho_i$.

\subsubsection{Nonaffine Forward Propagation}

For each non-empty partition $\mathcal{R}_{k,i} \neq \emptyset$, we compute:

\begin{equation}
\mathcal{R}^i_{k+1} = f_i(\mathcal{R}_{k,i} \times \mathcal{U}_k), 
\label{eq:propagation}
\end{equation}

where $f_i(\mathcal{R}_{k,i} \times \mathcal{U}_k)$ denotes the image of the Cartesian product under the nonaffine map $f_i$.

\begin{remark}
The computation of \eqref{eq:propagation} employs polynomial enclosure techniques combined with interval arithmetic to obtain tight over-approximations while maintaining the HPZ representation.
\end{remark}

\subsubsection{Set Union Operation}

The reachable set at time $k+1$ is obtained by  the union operation for HPZ \ref{union}:

\begin{equation}
\mathcal{R}_{k+1} = \bigcup_{i \in S} \mathcal{R}^i_{k+1}.
\label{eq:union}
\end{equation}
And we calculate the next step with $\mathcal{R}_{k+1}$.
\begin{algorithm}[t]
\caption{Reachability Analysis for Hybrid Nonaffine Systems}
\label{alg:reachability}
\begin{algorithmic}[1]
\Require Initial set $\mathcal{X}_0$, input sets $\{\mathcal{U}_k\}_{k=0}^{N-1}$, horizon $N$
\Ensure Reachable sets $\{\mathcal{R}_k\}_{k=0}^{N}$
\State $\mathcal{R}_0 = \mathcal{X}_0$
\For{$k = 0, 1, \ldots, N-1$}
    \For{$i \in S$}
        \State $\mathcal{R}_{k,i} = \text{halfspace}(\mathcal{R}_k, L_i, \rho_i)$ 
        \If{$\mathcal{R}_{k,i} \neq \emptyset$}
            \State $\mathcal{R}^i_{k+1} = f_i(\mathcal{R}_{k,i} \times \mathcal{U}_k)$
        \EndIf
    \EndFor
    \State $\mathcal{R}_{k+1} = \bigcup^s_{i =1} \mathcal{R}^i_{k+1}$
\EndFor
\State \Return $\{\mathcal{R}_0, \mathcal{R}_1, \ldots, \mathcal{R}_N\}$
\end{algorithmic}
\end{algorithm}




\section{Numerical examples}\label{sec-numericalEx}


%

All numerical experiments presented in this work were conducted using MATLAB R2024b on Microsoft Windows. The computational platform features an Intel Core i7-12800HX processor (12th Gen) with a base frequency of 4000 MHz and 16 GB of physical memory, providing sufficient computational resources for the reachability analysis of hybrid systems.

The implementation leverages the COntinuous Reachability Analyzer (CORA) toolbox version 2025~\cite{althoff2015introduction}, which provides comprehensive support for reachability analysis of hybrid systems with various set representations including zonotopes, polynomial zonotopes, and the newly introduced hybrid polynomial zonotopes.


To demonstrate the effectiveness of hybrid polynomial zonotopes in handling nonlinear hybrid systems, we consider a discrete-time piecewise nonaffine (PWNA) system with quadratic dynamics in each mode. The system dynamics are given by:
\begin{equation}\label{eq:pwna-dynamics}
    x_{k+1} = \begin{cases}
        f_1(x_k) = x_k^T M_1 x_k + A_1 x_k + d, & \text{if } x_{1,k} \leq 0, \\
        f_2(x_k) = x_k^T M_2 x_k + A_2 x_k + d, & \text{otherwise},
    \end{cases}
\end{equation}
where $x_k = [x_{1,k}, x_{2,k}]^T \in \mathbb{R}^2$ denotes the state at time step $k$. The system parameters are:
\begin{align}
&M_1 = \begin{bmatrix}
    0.017 & -0.0028 \\
    0 & 0.017
\end{bmatrix}, \quad
M_2 = \begin{bmatrix}
    -0.1 & 0 \\
    0 & -0.1
\end{bmatrix}, \label{eq:quadratic-matrices}\\
&A_1 = \begin{bmatrix}
    0.75 & 0.25 \\
    -0.25 & 0.75
\end{bmatrix},
A_2 = \begin{bmatrix}
    0.75 & -0.25 \\
    0.25 & 0.75
\end{bmatrix}, d = \begin{bmatrix}
    0.25 \\
    -0.5
\end{bmatrix}. \label{eq:linear-matrices} 
\end{align}

This system exhibits several challenging characteristics for reachability analysis:
\begin{enumerate}
    \item Quadratic Nonlinearity: Each mode contains a quadratic term $x^T M_i x$ that introduces state-dependent nonlinear behavior, causing the reachable sets to deform in complex ways.
    \item Mode Switching: The linear guard condition $\mathcal{G} = \{x \in \mathbb{R}^2 : x_1 \leq 0\}$ partitions the state space, requiring accurate computation of guard intersections.
    \item Coupled Dynamics: The matrices $A_1$ and $A_2$ represent rotational transformations with scaling, while $M_1$ introduces expansion and $M_2$ introduces contraction, creating rich dynamical behavior.
\end{enumerate}

For the reachability analysis, we initialize the system with a zonotopic initial set:
\begin{equation}\label{eq:initial-set}
    \mathcal{Z}_0 = \left\langle c_0, G_0 \right\rangle = \left\langle \begin{bmatrix}
        -0.201 \\
        0.96
    \end{bmatrix}, \begin{bmatrix}
        0.2 & 0 \\
        0 & 0.2
    \end{bmatrix} \right\rangle,
\end{equation}
representing a square region with center $c_0$ and half-edge length 0.2.
\begin{figure}[htbp]
    \centering
    \includegraphics[width=0.48\textwidth]{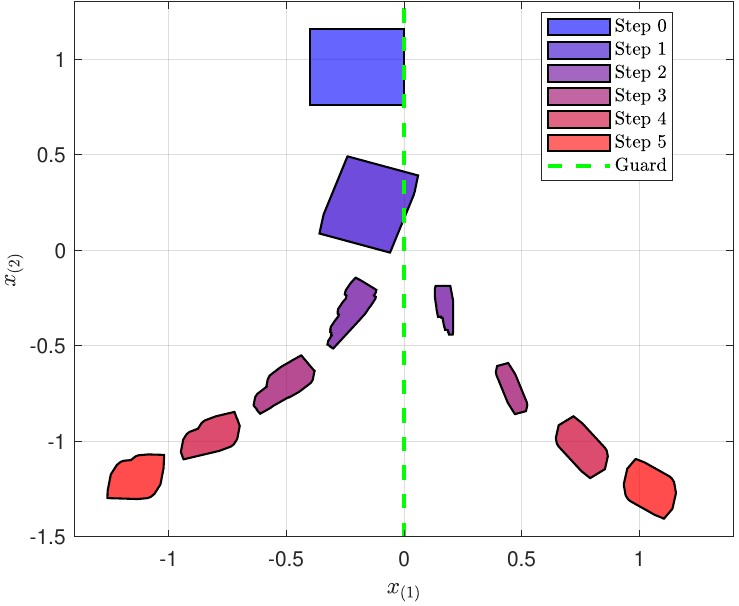}
    \caption{Reachable set evolution for PWNA system over 5 time steps. Color progression tracks temporal evolution; green line marks switching boundary at $x_1 = 0$.}
    \label{fig:pwna_reachable_sets}
\end{figure}

Figure~\ref{fig:pwna_reachable_sets} demonstrates HPZ's ability to handle mode transitions and nonlinear dynamics efficiently. The algorithm precisely captures intersection operations at the guard boundary without excessive over-approximation. Mode 2's contractive quadratic dynamics produce complex nonconvex geometries that traditional convex representations would poorly approximate. The complete analysis requires only 2.12 seconds, confirming HPZ's computational efficiency for hybrid nonlinear systems.

\section{Conclusion}
\label{sec-conclusion}

This paper introduces Hybrid Polynomial Zonotopes, a novel set representation that successfully enables precise and exact representation of reachable sets for hybrid nonlinear nonaffine systems. HPZ addresses the limitation of existing zonotope methods, which cannot accurately describe inherently non-convex reachable sets within a unified representation framework, instead relying on unions of convex components. HPZ combines the mode-dependent structure of hybrid zonotopes with polynomial expressiveness, providing a closed-form representation under essential set operations including Minkowski sum, linear transformation, intersection, and union. We developed comprehensive mathematical frameworks and computational algorithms for HPZ operations, demonstrating superior accuracy and efficiency compared to existing approaches. Future work will explore HPZ applications in controller synthesis and extend the framework to stochastic hybrid systems.

\bibliographystyle{IEEEtran} 
\bibliography{BibTex_2025}

\begin{thebibliography}{10}
\providecommand{\url}[1]{#1}
\csname url@samestyle\endcsname
\providecommand{\newblock}{\relax}
\providecommand{\bibinfo}[2]{#2}
\providecommand{\BIBentrySTDinterwordspacing}{\spaceskip=0pt\relax}
\providecommand{\BIBentryALTinterwordstretchfactor}{4}
\providecommand{\BIBentryALTinterwordspacing}{\spaceskip=\fontdimen2\font plus
\BIBentryALTinterwordstretchfactor\fontdimen3\font minus \fontdimen4\font\relax}
\providecommand{\BIBforeignlanguage}[2]{{%
\expandafter\ifx\csname l@#1\endcsname\relax
\typeout{** WARNING: IEEEtran.bst: No hyphenation pattern has been}%
\typeout{** loaded for the language `#1'. Using the pattern for}%
\typeout{** the default language instead.}%
\else
\language=\csname l@#1\endcsname
\fi
#2}}
\providecommand{\BIBdecl}{\relax}
\BIBdecl

\bibitem{mcclamroch2002performance}
N.~H. Mcclamroch and I.~Kolmanovsky, ``Performance benefits of hybrid control design for linear and nonlinear systems,'' \emph{Proceedings of the IEEE}, vol.~88, no.~7, pp. 1083--1096, 2002.

\bibitem{borquez2023hamilton}
J.~Borquez, S.~Peng, Y.~Chen, Q.~Nguyen, and S.~Bansal, ``Hamilton-jacobi reachability analysis for hybrid systems with controlled and forced transitions,'' \emph{arXiv preprint arXiv:2309.10893}, 2023.

\bibitem{prieur2007hybrid}
C.~Prieur, R.~Goebel, and A.~R. Teel, ``Hybrid feedback control and robust stabilization of nonlinear systems,'' \emph{IEEE Transactions on Automatic Control}, vol.~52, no.~11, pp. 2103--2117, 2007.

\bibitem{sun2011stability}
Z.~Sun and S.~S. Ge, ``Stability theory of switched dynamical systems,'' 2011.

\bibitem{hespanha2004stability}
J.~P. Hespanha and A.~S. Morse, ``Stability of switched systems with average dwell-time,'' \emph{IEEE Transactions on Automatic Control}, vol.~44, no.~11, pp. 2091--2095, 1999.

\bibitem{liberzon2003switching}
D.~Liberzon, \emph{Switching in Systems and Control}.\hskip 1em plus 0.5em minus 0.4em\relax Boston: Birkh{\"a}user, 2003.

\bibitem{schwartz2005minimum}
C.~A. Schwartz and E.~Maben, ``A minimum energy approach to switching control for mechanical systems,'' in \emph{Control using logic-based switching}.\hskip 1em plus 0.5em minus 0.4em\relax Springer, 2005, pp. 142--150.

\bibitem{long2022robust}
L.~Long, F.~Wang, and Z.~Chen, ``Robust adaptive dynamic event-triggered control of switched nonlinear systems,'' \emph{IEEE Transactions on Automatic Control}, vol.~68, no.~8, pp. 4873--4887, 2022.

\bibitem{long2021dynamic}
L.~Long and F.~Wang, ``Dynamic event-triggered adaptive nn control for switched uncertain nonlinear systems,'' \emph{IEEE Transactions on Cybernetics}, vol.~53, no.~2, pp. 988--999, 2021.

\bibitem{bird2023hybrid}
T.~J. Bird, H.~C. Pangborn, N.~Jain, and J.~P. Koeln, ``Hybrid zonotopes: A new set representation for reachability analysis of mixed logical dynamical systems,'' \emph{Automatica}, vol. 154, p. 111107, 2023.

\bibitem{althoff2010reachability}
M.~Althoff, ``Reachability analysis and its application to the safety assessment of autonomous cars,'' Ph.D. dissertation, Technische Universit{\"a}t M{\"u}nchen, 2010.

\bibitem{nikolaou2003linear}
M.~Nikolaou, ``Linear control of nonlinear processes: recent developments and future directions,'' \emph{Computers \& Chemical Engineering}, vol.~27, no. 8-9, pp. 1043--1059, 2003.

\bibitem{asarin2000reachability}
E.~Asarin, O.~Bournez, T.~Dang, and O.~Maler, ``Approximate reachability analysis of piecewise-linear dynamical systems,'' in \emph{International Workshop on Hybrid Systems: Computation and Control}.\hskip 1em plus 0.5em minus 0.4em\relax Springer, 2000, pp. 20--31.

\bibitem{emara1996nonlinear}
H.~E. Emara-Shabaik, ``Nonlinear systems modeling \& identification using higher order statistics/polyspectra,'' in \emph{Control and Dynamic Systems}.\hskip 1em plus 0.5em minus 0.4em\relax Elsevier, 1996, vol.~76, pp. 289--322.

\bibitem{nonlinear_approximation}
A.~Benallou, D.~E. Seborg, and D.~A. Mellichamp, ``Dynamic compartmental models for separation processes,'' \emph{AIChE Journal}, vol.~32, no.~7, pp. 1067--1078, 1986.

\bibitem{scott2016constrained}
J.~K. Scott, D.~M. Raimondo, G.~R. Marseglia, and R.~D. Braatz, ``Constrained zonotopes: A new tool for set-based estimation and fault detection,'' \emph{Automatica}, vol.~69, pp. 126--136, 2016.

\bibitem{kochdumper2020sparse}
N.~Kochdumper and M.~Althoff, ``Sparse polynomial zonotopes: A novel set representation for reachability analysis,'' \emph{IEEE Transactions on Automatic Control}, vol.~66, no.~9, pp. 4043--4058, 2020.

\bibitem{kochdumper2023}
------, ``Constrained polynomial zonotopes,'' \emph{Acta Informatica}, vol.~60, no.~3, pp. 279--316, 2023.

\bibitem{zhang2025data}
Z.~Zhang, M.~U.~B. Niazi, M.~S. Chong, K.~H. Johansson, and A.~Alanwar, ``Data-driven nonconvex reachability analysis using exact multiplication,'' \emph{arXiv preprint arXiv:2504.02147}, 2025.

\bibitem{hadjiloizou2024formal}
L.~Hadjiloizou, F.~J. Jiang, A.~Alanwar, and K.~H. Johansson, ``Formal verification of linear temporal logic specifications using hybrid zonotope-based reachability analysis,'' in \emph{2024 European Control Conference (ECC)}.\hskip 1em plus 0.5em minus 0.4em\relax IEEE, 2024, pp. 579--584.

\bibitem{koeln2024zonolab}
J.~Koeln, T.~J. Bird, J.~Siefert, J.~Ruths, H.~C. Pangborn, and N.~Jain, ``zonolab: A matlab toolbox for set-based control systems analysis using hybrid zonotopes,'' in \emph{2024 American Control Conference (ACC)}.\hskip 1em plus 0.5em minus 0.4em\relax IEEE, 2024, pp. 2513--2520.

\bibitem{siefert2025reachability}
J.~A. Siefert, T.~J. Bird, A.~F. Thompson, J.~J. Glunt, J.~P. Koeln, N.~Jain, and H.~C. Pangborn, ``Reachability analysis using hybrid zonotopes and functional decomposition,'' \emph{IEEE Transactions on Automatic Control}, 2025.

\bibitem{bird2021unions}
T.~J. Bird and N.~Jain, ``Unions and complements of hybrid zonotopes,'' \emph{IEEE Control Systems Letters}, vol.~6, pp. 1778--1783, 2021.

\bibitem{xie2025data}
P.~Xie, J.~Betz, D.~M. Raimondo, and A.~Alanwar, ``Data-driven reachability analysis for piecewise affine systems,'' \emph{arXiv preprint arXiv:2504.04362}, 2025.

\bibitem{althoff2015introduction}
M.~Althoff, ``An introduction to cora 2015,'' in \emph{Proc. of the workshop on applied verification for continuous and hybrid systems}.\hskip 1em plus 0.5em minus 0.4em\relax Citeseer, 2015, pp. 120--151.

\end{thebibliography}

\end{document}